%% file: main.tex
\documentclass[5p,times]{elsarticle}

\usepackage{amssymb}
\usepackage{graphicx} 
\usepackage{array} 
\usepackage{verbatim} 
\usepackage{amsmath}
\usepackage{amsfonts}
\usepackage{amssymb}
\usepackage{amsthm}
\usepackage{dsfont}

\usepackage[lined,algonl,boxed]{algorithm2e}

\newtheorem{theorem}{Theorem}
\newtheorem{lemma}{Lemma}

\newtheorem{definition}{Definition}

\newtheorem{problem}{Problem}

\journal{Future Generation Computer Systems}

\begin{document}

\begin{frontmatter}



\title{Building Stable Off-chain Payment Networks }


\author{MohammadAmin Fazli, Seyed Moeen Nehzati, MohammadAmin Salarkia}

\address{Sharif University of Technology, Tehran, Iran.}

\begin{abstract}
Payment channel is a protocol which allows cryptocurrency users to route multiple transactions through  network without committing them to the main blockchain network (mainnet). This ability makes them the most prominent solution to blockchains' scalability problem. Each modification of payment channels requires a transaction on the mainnet and therefore, big transaction fees. In this paper, we assume that a set of payment transactions are given (batch or online) and we study the problem of scheduling modificiations on payment channels to route all of the transactions with minimum modification cost. 

We investigate two cost models for aforementioned problem: the step cost function in which every channel modification has a constant cost and the linear cost function in which modification costs are proportional to the amount of change. For the step cost function model, we prove impossibility results for both batch and online case. Moreover, some heuristic methods for the batch case are presented and compared. For the linear cost we propose a polynomial time algorithm using linear programming for the batch case.
\end{abstract}


\end{frontmatter}


\section{Introduction} \label{introduction}
\input{introduction}

\input{problem}
\input{optimum}
\input{heuristics}
\input{conclusion}




\bibliographystyle{model1-num-names}
\bibliography{refs}







\end{document}

%% file: introduction.tex
Recently, blockchain-based decentralized cryptocurrencies like Bitcoin, Ethereum, and Zcash have got much popularity and attention. In these decentralized cryptocurrencies, transactions don't need to be managed by third parties \cite{Zhou2020survey}, but the problem with these currencies is their performance and scalability. Since their transaction throughput is bounded by maximum block size in the chain, they are not scalable for the high number of transactions \cite{croman2016scaling}. For example, for bitcoin, the best known bound for average transaction throughput is seven transactions per second \cite{cryptoeprint:2019:416}. This value is an order of magnitude lesser than the transaction rate in everyday payments, such as Visa cards. Therefore, it will not be possible to use bitcoin in people's daily life. Moreover, this low transaction rate has resulted in competition between users to finalize their transactions faster, which resulted in higher transaction fees.

Many solutions have been provided to solve the scalability problem in blockchain networks. Kim et al. \cite{Kim2018survey} and Zhou et al. \cite{Zhou2020survey} provide surveys on these solutions. Some of these solutions are on-chain methods, which try to solve the problem within the main blockchain, for example, by modifying some elements in the chain to increase the maximum block size. Other solutions are non-on-chain methods such as off-chain, side-chain, cross-chain, etc. These solutions use another network or chain to improve scalability outside the blockchain. One of these non-on-chain solutions, which was very promising in recent years, is payment channels. With this solution, some transactions will transfer to a temporary off-chain network to reduce the number of transactions on the main chain. Lightning network for Bitcoin and Raiden network for Ethereum are examples of such payment channel networks.

A payment channel is a cryptocurrency transaction on blockchain that deposits some money on the blockchain for exchange between two specific users \cite{sivaraman2020high}. For instance, Alice can deposit 1 BTC into a multi-signature transaction (meaning that the signatures of both parties is needed for withdrawing) for sending to Bob. Setting up this transaction follows the common mechanism of transactions on blockchain and needs to be committed on main net. From now on, Alice can send money up to 1 BTC to Bob as many times as she wants by sending signed transctions from the deposit account. Bob can close payment channel by signing the last transaction and broadcasting it into the main net. This whole process just needs two commits on blockchain. The first one for setting up the payment channel and the second one for closing it.

However, payment channels can only handle transactions between two users. So, it won't solve the scalability problem if everyone has to set up a payment channel to another user to provide fast payments to him. For this purpose, offchain networks of payment channels are proposed \cite{poon2016bitcoin}. With these networks, This is where the Lightning Network comes up. With this network, every two users can make transactions if there is a path of payment channels between them.

In this paper, we target the problem of building stable payment networks. In this problem, every channel has a specific deposit. That means it can route transactions with values less than its deposit. We call this deposit channel capacity, and it decreases with routing every transaction. Also, we assume that every node in this network has a finite capital which means the sum of channels' capacities sourcing from that node must be less than its capital. As a result, channels can't have infinite capacity. Every change in channel capacities needs a transaction commit on the mainnet, and it is costly. In the target problem of this paper, a network of payment channels and a set of transactions on the network are given. Then, we want to schedule capacity changes on payment channels to minimize the cost of commitments on the mainnet.

To evaluate different scheduling algorithms, we investigate two cost models. The first one is the linear cost model. In this case, the change cost of a channel's capacity is linearly dependent on the value of the change. The second one is the step cost model, which means the cost of capacity change is constant. Moreover, we check out two cases for transactions: the batch case and the online case. In the batch case, we assume that all the input transactions are given at time zero. In the online case, we suppose that at each time, we only know the next transaction. This paper proposes a polynomial-time algorithm to find the optimum solution for the linear cost model and the batch input case. For the step cost model and the batch input case, we prove the NP-completeness of the solution. For the step cost model and the online input case, we show no c-competitive algorithm such that \(c \leq \delta\), where \(\delta\) is the maximum degree of nodes in our network. Finally, we propose some heuristic algorithms for the step cost model and the batch input case and compare their performance in a simulation.

This section will explain payment channels and the payment network concepts in more detail and then survey the literature. In the next section, we will propose the formal statement of our problem. After that, we will state the aforementioned theoretical results in finding optimum solutions. Next, we will present our report on heuristic algorithms and their performance. We will conclude this paper by reviewing the achievements made in this paper and the possible future works.

\subsection{Payment channels}
The Payment channel can be mono-directional or bi-directional. First, let's see the mono-directional form. Imagine the case that Alice wants to open a payment channel with Bob with the capacity of 1 BTC. She has to deposit this money in a multi-signature address. All the transactions on this address need the signature of both Alice and Bob to be performed. If Alice wants to pay 0.1 BTC to Bob, she creates a transaction that gives 0.9 BTC from that specific address to Alice and gives 0.1 BTC to Bob. She signs this transaction and gives it to Bob. Sometime later, when Alice wants to pay another 0.2 BTC to Bob, she writes a transaction that gives 0.7 BTC to herself and gives 0.3 BTC to Bob. Again, she signs this transaction and gives it to Bob. Whenever Bob wants to withdraw his money, he can sign the transaction (complete the needed signatures), and then, this transaction will broadcast on the main blockchain and close the channel. Note that he can't commit any other transaction after Bob has committed the last transaction. In Figure \ref{paymentchannel}, you can see this scenario.

The problem here is that if Alice wants her 0.7 BTC, she depends on Bob's action. Thus, there must be a deadline for Bob to withdraw the money; otherwise, all the money will turn back to Alice. A transaction that gives all the money to Alice after some time will do this aim. Before Alice funds the address, Bob will write and sign this transaction and gives it to Alice. Then, Alice will fund the channel.

However, mono-directional payment channels do not look promising because they only allow one-way payments. So there is a need for bi-directional channels. In a bi-directional channel, both Alice and Bob create guarantee transactions and give them to each other. After that, they both fund the channel, for example, by \(V_a\) and \(V_b\) values. Now, if Alice wants to pay Bob \(v_1\), she can create a transaction that gives herself \(V_a - v_1\) and gives Bob \(V_b + v_1\), sign it and give it to Bob. Then, if Bob wants to give Alice \(v_2\), he can write a transaction which gives Alice \(V_a - v_1 + v_2 \) and gives himself \(V_b + v_1 - v_2\) and so on.

Since these channels (both mono-directional and bi-directional) only need two transactions on the main blockchain (one for opening and one for closing), they look very promising for solving the scalability problem of blockchain-based payments.

\begin{figure*}[h] \label{paymentchannel}
\centering
\includegraphics[width=\textwidth]{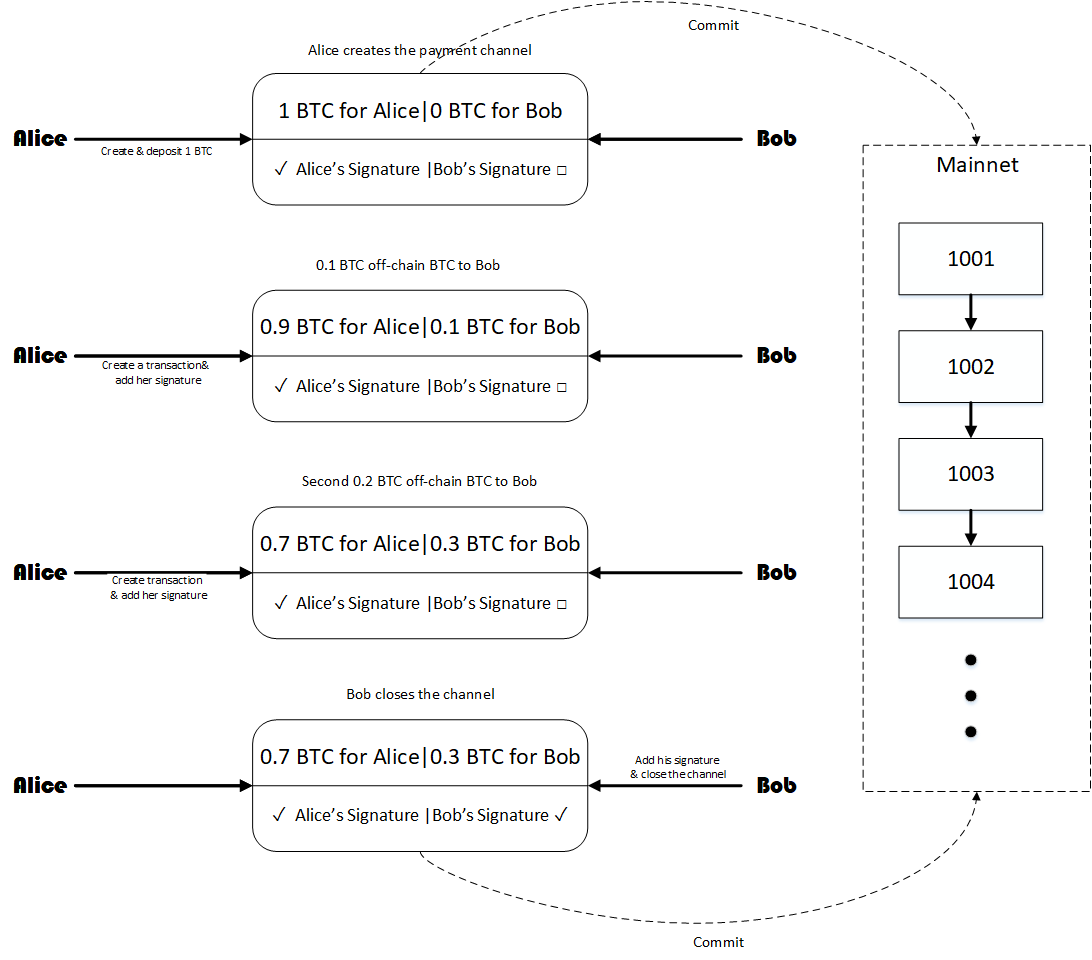}
\caption{The Payment Channel Mechanism}
\end{figure*}

\subsection{Payment Networks}
Although payment channels need just two transactions on the blockchain, if everyone has to create a payment channel with another one, it doesn't solve the scalability problem. Creating these many channels and closing them needs a lot of time and locks many coins only for future payments. Here the solution of payment networks comes up. Payment networks are graphs in which users are the nodes, and payment channels between them are edges. If some user Alice wants to pay some money to user Carol and there is a path between them in the graph, they can do their transaction only if the value of the transaction is less than the channels' capacities in this path (See Figure \ref{paymentnetwork}). 

\begin{figure*}[h] \label{paymentnetwork}
\centering
\includegraphics[width=0.7\textwidth]{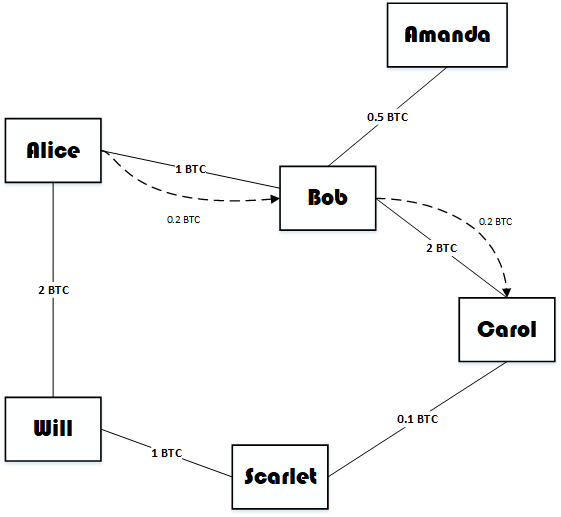}
\caption{A Payment Network; Alice wants to pay 0.2 bitcoin to Carol. This payment is impossible through the Alice-Will-Scarlet-Carol path because the channel between Scarlet and Carol has not enough capacity. However the Alice-Bob-Carlo path makes it possible.  First, Alice pays 0.2 BTC to Bob through her channel to him. Then Bob pays this money to Carol through his channel to him. }
\end{figure*}

One of the main problems with payment networks is how to incentivize intermediate nodes to route other nodes' transactions. This problem can be solved by giving a routing fee to all users in the path. With routing fees, the users in the network tend to create more channels with higher capacities to achieve more fees.

Another problem with this network is that intermediate users can steal the routing coins. Hash-locked contracts can prevent this. Hash locked contracts are transactions that can only be finalized by providing a specific proof. For instance, assume Alice wants to pay Carol, and there is a path of length 2 (Alice-Bob-Carol) between them. First, Carol has to generate a private key and a public key and give the public key to Alice. Then, Alice makes a transaction with Bob, which needs the private key to be done. Also, Bob makes a similar transaction with Carol. Now Carol provides the private key and does the transaction between herself and Bob. After that, Bob will know the private key, and he can also give the private key and do the transaction between herself and Alice.

The Lightning Network is the payment channel network for bitcoin. The Raiden Network is the Ethereum version. It is very similar to the Lightning Network except that it supports different types of transactions while the lightning network supports only bitcoin transactions. 

\subsection{Related Work}
Kim et al. \cite{Kim2018survey}, and Zhou et al. \cite{Zhou2020survey} provide surveys on solutions to the scalability problem on blockchain-based payments. As one of the most promising solutions, researchers have proposed various payment channel protocols (networks) which present how to create a payment channel \cite{avarikioti2020cerberus,decker2015fast,poon2016bitcoin,green2017bolt}. However, the current work is independent of channel protocols and applies to all such solutions.  The Lightning Network, which is the most famous payment channel, was first introduced in \cite{poon2016bitcoin}. 

Most of the optimization and algorithmic research works done by a team resided in the ETH university.  Avarikioti et al. \cite{avarikioti2018payment} study the optimal structure of the channels' network and the allocation of charges to them to maximize the profit of the payment service provider (PSP). In their other work \cite{avarikioti2018algorithmic}, proposes three optimization problems in the context of payment network design. All of those problems are variants of the {\bf General Network Design} problem. In this problem, a set of transactions, a capital, and a profit are given. The problem asks to design a network and a strategy to route the transaction through the network.  The first target problem gets a set of transactions and asks whether it is possible to route all of them from a single bidirectional channel between two users. The second problem gets a set of transactions and asks to return a graph with the minimum possible capital. The third problem receives a set of transactions, a graph, and a capital and asks whether there is a capital assignment to the graph's edges that can route all the transactions. They present hardness results and approximation algorithms for the batch case of these three problems. For the online case, they give an impossibility result for the first problem and a competitive algorithm for the second problem. 

The ETH team has also focused on some other aspects of payment channels. For example, in \cite{ Avarikioti2020RideTL}, they explore the game theory of payment networks. They assume that the users in the network are selfish players, and they compete to earn more. Then study the Nash equilibrium of the network creation game, which are the topologies that emerge from their interactions. For investigating the topology, they consider betweenness and closeness centrality as the central concept. At last, they determine the social optima for the network topology.  

There are some other works targetting payment channels and the Lightning Network. Most of these works are unrelated to this paper, but some of them are surveyed as examples.  Sivaraman et al. \cite{sivaraman2020high} Claim that shortest-path routing is not a good idea for routing transactions on payment networks (because of the limited capacity of channels, some important channels may get exhausted, and the system will stop working). Then, it provides a new routing solution of transactions to achieve high throughput.
Prihodko et al. \cite{prihodko2016flare} provide another routing algorithm of transactions in the Lightning network called Flare. Seres et al. \cite{seres2020topological} analyze the effects of the Lightning network's topology on its security and its robustness against failures and attacks. 
Herrera et al. \cite{Herrera2019Hiding} mention that the precise balance of the pair of nodes in payment networks has to be kept secret due to privacy problems. Then, they investigate the problem of hiding this balance and having working routing algorithms simultaneously and proposes an attack to reveal the balance of a channel.

%% file: problem.tex
 \section{The Problem Statement} \label{problem}
Here we will go with the formal definition of our problem. We are given a payment network \(G = (V, E)\), Where V stands for a set of nodes that are cryptocurrency users (more precisely, addresses) and can have payment transactions between themselves. \(E\) is the set of \(G\)'s payment channels. Each \(e \in E\), is a directed pair \(e=(u,v), u,v \in V\). These channels can be used for money transfers. The capacity of each channel is defined by a function \(c_e: T \rightarrow \mathds{R}^+\), where \(T\) is the set of times. Without loss of generality, we can assume that times are positive integers from 1 to \(n\), which means: \(T=\{1,2,...,n\}\).
\begin{definition}
  For every channel \(e \in E\) in graph \(G=(V,E)\) {\bf the capacity of channel} at time \(t\) is represented by \(c_e(t)\), where \(c_e: T \rightarrow \mathds{R}^+\) is the capacity function and \(T=\{i| 1 \leq i \leq n\}\) and \(n\) is the number of time steps. 
\end{definition}

Now, suppose that a set of transactions \(\Gamma\) is given. Without loss of generality we can assume that in each time \(i \in T\) we have exactly one input transaction \(t_i \in \Gamma\). Thus, the number of input transactions is equal to n. Each transaction is defined like this: \(t_i = < s_i, d_i, p_i, v_i >\), where \(s_i, d_i \in V\) are the source and the destination of the transactions, \(p_i\) is a path (an ordered set of channels) in \(G\) between \(s_i\) and \(d_i\) and finally, \(v_i\) is the value of transaction \(t_i\). In practical settings most of the times this path is the shortest path between \(s_i\) and \(d_i\) in \(G\) (but sometimes it is not a good idea to use the shortest path, See \cite{sivaraman2020high}). The transaction \(t_i\) can be executed in \(G\) if:
\[c_e(i) \geq v_i \; ; \; \forall e \in p_i\]. 

In this paper, we want to schedule the modifications of \(G\)'s channels' capacity in a way that can route all the input transactions with minimum cost. Each channel's capacity is constrained by its connected nodes' capital. We define the capital of each node \(u \in V\) as a function \(C_u: T \rightarrow \mathds{R}^+\), where \(C_u(t)\) represents \(u\)'s capital at time \(t\).  
\begin{definition}
For every node \(u \in V\) in graph \(G=(V,E)\), {\bf the capital of node} in time step \(t\) is represented by \(C_u(t)\), where \(C_u: T \rightarrow \mathds{R}^+\) is the capital function and \(T=\{i| 1 \leq i \leq n\}\) and \(n\) is the number of time steps. 
\end{definition}
We assume that no node can change its capital while processing payment transactions. Therefore, \(C_u\) can only be increased or decreased by transactions.  \(C_u\) decreases when \(u\) is the source of a transaction, i.e., \(u\) wants to pay some amount of money to someone. In the same way, \(C_u\) increases when \(u\) is the destination of a transaction, i.e., someone wants to pay \(u\) some money.
\begin{lemma}
With set of transactions \(\Gamma=\{<s_i, d_i, p_i, v_i> | 1 \leq i \leq n\}\), the capital of node \(u\) at time \(t\) can be represented in this way:
\[C_u(t) = C_u(0) - \sum_{j \leq t, s_j=u}v_j + \sum_{j \leq t, d_j=u}v_j\]
\end{lemma}

Our desired scheduling plan consists of a scheme for each channel. The scheme for each channel \(e\) is defined as a sequence of channel modifications \(f_e = \{(\tau_i^e, \lambda_i^e)\}_{i=0}^{k_e}\) , where \(k_e\) is the number of modifications by our plan on channel \(e\), \(\tau_i^e \in T\) is the scheduled time for \(i\)th modification (with \(\tau_0^e = 0\)) and \(\lambda_i^e\) is the new amount which \(e\) gets after this modification. Note that its initial value \(\lambda_0^e\) equals to the initial capcity of channel \(c_e(0)\). 
\begin{definition}
{\bf A channels modification scheme} \(f_e = \{(\tau_i^e, \lambda_i^e)\}_{i=0}^{k_e}\) for edge \(e\) is defined in a way that \(\lambda_i^e\) is the value of capacity of channel \(e\) between times \(\tau_i^e\) and \(\tau_{i+1}^e\). We can represent this definition by this formula:
\[\tau_i^e \leq t < \tau_{i+1}^e \rightarrow c_e(t) = \lambda_i^e - \sum_{\tau_j^e \leq j < \tau_{j+1}^e, e \in p_j}v_j + \sum_{\tau_j^e \leq j < \tau_{j+1}^e, e \in \bar{p_j}} v_j,\]
where \(\bar{p_j}\) is the path \(p_j\) in reverse direction.
\end{definition}
This paper defines the problem of designing a plan for scheduling the input payment transactions on a given payment network as an optimization problem with a set of constraints and an objective function. The constraints are defined based on the fact that each node \(u\) must provide the capacity of its outgoing channels, and the sum of its outgoing channels capacities should be less than its capital. more formally: \(\sum_{e \in O(u)} c_e(t) \leq C_u(t)\), where \(O(u)\) is the set of edges sourced by \(u\), i.e. \(O(u) = \{(u,v)| (u,v) \in E\}\).

For the optimization objective, we consider two different functions, the linear cost function, and the step cost function. The linear cost function measures the cost of each channel modification by its amount of change in its capital.
\begin{definition}
For a given graph \(G\) and a set of transactions \(\Gamma\) and a channels modifications scheme \(f_e = \{(\tau_i^e, \lambda_i^e)\}_{i=0}^{k_e}\), we can define {\bf the linear cost function} of channels modifications as follows:
\[LC(G, \Gamma) = \sum_{e \in E(G)}\sum_{i=1}^{k_e} | \lambda_i^e - c_e(\tau_i^e) |.\]
\end{definition} 
The step cost function computes the cost as the number of changes to the capacities.
\begin{definition}
For given graph \(G\) and set of transactions \(\Gamma\) and channels modifications scheme \(f_e = \{(\tau_i^e, \lambda_i^e)\}_{i=0}^{k_e}\), we can define {\bf the step cost function} of the channels modifications as follows:
\[SC(G, \Gamma) = \sum_{e \in E(G)} k_e\]
\end{definition}
Finally, with these definitions, we can propose the formal statement of our problem.
\begin{problem}
{\bf Scheduling capacity changes on payment networks with linear cost}\\
Given a graph \(G(V, E)\) and a set of transactions \(\Gamma=\{<s_i, d_i, p_i, v_i> | 1 \leq i \leq n\}\), we want to propose \(f_e = \{(\tau_i^e, \lambda_i^e)\}_{i=0}^{k_e}\) for every \(e \in E\) in a way that minimizes \(LC(G, \Gamma)\) and meets these constraints:
	$$ c_e(i) \geq v_i ; \forall 1 \leq i \leq n , \forall e \in p_i, $$
	$$ \sum_{e \in O(u)} c_e(t) \leq C_u(t), $$
where \(c_e\) is capacity function of edge \(e\), \(C_u\) is capital function of node \(u\) and \(O(u)\) is set of outgoing channels from \(u\).
\label{problem-linear}
\end{problem}
\begin{problem}
{\bf Scheduling capacity changes on payment network with step cost}\\
Given a graph \(G(V, E)\) and a set of transactions \(\Gamma=\{<s_i, d_i, p_i, v_i> | 1 \leq i \leq n\}\), we want to propose \(f_e = \{(\tau_i^e, \lambda_i^e)\}_{i=0}^{k_e}\) for every \(e \in E\) in a way that minimizes \(SC(G, \Gamma)\) and meets these constraints:
	$$ c_e(i) \geq v_i ; \forall 1 \leq i \leq n , \forall e \in p_i, $$
	$$ \sum_{e \in O(u)} c_e(t) \leq C_u(t), $$
where \(c_e\) is capacity function of edge \(e\), \(C_u\) is capital function of node \(u\) and \(O(u)\) is set of outgoing channels from \(u\).
\label{problem-step}
\end{problem}
We face these two problems with two types for \(\Gamma\), batch and online. In the batch form, all \(t_i \in \Gamma\) are given at first, But, in the online form in every time step \(i\) we just have \(t_j \in \Gamma; 1 \leq j \leq i\) and we don't know the future transactions. 

%% file: optimum.tex
\section{Theoretical Results}
In this section, we are going to investigate the optimal solution to the problems stated in Section \ref{problem}. For the problem \ref{problem-linear}, which used the linear cost function, we propose a linear programming algorithm for scheduling modifications in polynomial time. For the problem \ref{problem-step}, we investigate both batch and online cases. For the batch input case, we reduce a well-known fault caching problem to our problem. Moreover, as it is proved that fault caching is an NP-Complete algorithm, we can conclude that the problem \ref{problem-step} is NP-Complete either. For the online input case, we define the competitiveness, and then we prove that there are no c-competitive algorithms for our problem where $ c \leq \Delta $ and $ \Delta $ equals the maximum degree of nodes in the graph. For that purpose, for each algorithm, we devise a graph in which the cost of the algorithm will be greater than the optimal cost multiplied by $ \Delta $.

\subsection{Batch Input - Linear Cost}
The following theorem shows that this case of the problem is solvable in polynomial time by a linear programming. 
\begin{theorem}
The problem \ref{problem-linear} with batch input can be solved in polynomial time. 
\end{theorem}
\begin{proof}
The following LP can model the problem. The first condition ensures that the scheduling plan is executable at each time. The second condition guarantees changes will not set the capacity of edges to a negative value. The third constraint ensures that the payment network can route all the incoming transactions. Th final condition guarantees the validity of nodes' capitals. 
\begin{table}
$\begin{array}{llll@{}llll}
\text{minimize}  & \sum\limits_{e \in E} \sum\limits_{i = 1}^{k_e} |\lambda_i^e - c_e(\tau_i^e)| &\\\\
\text{subject to}& c_e(i) = \lambda_i^e - \sum\limits_{\substack{\tau_{i}^e \leq j < \tau_{i+1}^e \\ e \in p_j}} v_j + \sum\limits_{\substack{\tau_{i}^e \leq j < \tau_{i+1}^e \\ e \in \bar{p}_j}} v_j\;\;\;  &\forall i \in T, e \in E\\\\
                 &                                            \lambda_i^e \geq 0, &\forall i \in T, e \in E \\\\ &
c_e(t) \geq v_i &\forall i \in T, e \in p_i \\\\&
\sum\limits_{e \in O(u)}c_e(t) \leq C_u(t) &\forall i \in T, v \in V\\\\
\end{array}$
\end{table}
As all the constraints are met and the algorithm minimizes the cost function, it finds the optimal solution. This linear programming problem can be solved in polynomial time and gives the optimal solution. 

\end{proof}

\subsection{Batch Input - Step Cost}
In this subsection, we show that the problem of scheduling capacity changes on payment networks with the step cost function is NP-complete. To this aim, we use a reduction from another NP-complete problem called the Fault Caching problem (\cite{chrobak2012caching}).  
\begin{problem} {\bf Fault Caching problem}\\
Assume that we are given set of pages \(\{pg_0, pg_1, ..., pg_{k-1}\}\) with sizes \(\{size(pg_i)\}_0^k\). Assume that we have a cache of size \(W\). In each time we can store a set of pages in the cache if the sum of their sizes is less than \(W\). A series of page requests will come in an order \(\sigma_0, \sigma_1, ..., \sigma_{n-1}\) where \(\sigma_i \in \{pg_0, pg_1, ..., pg_{k-1}\}\)  and \(\sigma_i\) is requested at time \(i\). If at time \(i, \sigma_i\) is not available in cache we say that a fault happens. The Fault Caching problem asks for a replacement policy that makes responding to \(\sigma_i\)s with at most \(F\) faults.
\label{fault-caching}
\end{problem}

\begin{theorem}
The problem of scheduling capacity changes on payment networks with the step cost function is NP-complete.
\end{theorem}
\begin{proof}
We want to make a reduction from the Fault Caching problem to our target problem. To this aim, consider the graph \(G = (V, E)\) with \(V = \{v_1, v_2\}\) and \(E = \{e_1, e_2, ..., e_k\}\) where all \(e_i\)s are from \(v_1\) to \(v_2\) and \(C_{v_1}(0) = W\). Now we construct transactions \(\Gamma\) like this:
\[t_i = <v_1, v_2, e_{I(\sigma_i)}, size(\sigma_i)>,\]
where \(I\) gives index of the requested page.

We show that the answer to the Fault Caching problem is yes if and only if we can find a modification schedule for \(G\) and \(\Gamma\) with cost less than or equal to \(F\).

First, we prove the if part. Consider a solution of the Fault Caching problem like this: \(d_i\) pages that are dropped in time \(i\) and \(b_i\) set of pages that are brought to the cache in time \(i\). Also set \(b_0\) set of pages that are brought to the cache at time \(0\). we construct the solution \(f_e = \{\tau_i^e, \lambda_i^e\}_{i=0}^{k_e}\) in this way:
\begin{gather*}
	f_{e_i} = \{(\tau, 0)| \forall \tau, pg_i \in d_{\tau}\} \cup \{(\tau, size(pg_i)) | \forall \tau, pg_i \in b_{\tau}\} \\
	\sum_{i=0}^{n-1} (|d_i| + |b_i|) \leq F  \rightarrow \sum_{i=0}^{k-1} |f_{e_i}| \leq F
\end{gather*}

So the cost of this modification schedule is at most \(F\).

Secondly, we prove the only if part. Suppose we have solution \(f_e\) for the scheduling problem. For each time step \(\tau_i^e\)that e is not same as \(e_{I(\sigma_i)}\) we change \(\tau_i^e\) to \(\tau_{i}^e + 1\) if \(\tau_i^e + 1 \notin f_e\) and remove \((\tau_i^e, \lambda_i^e)\) from \(f_e\) if \(\tau_i^e + 1 \in f_e\). (we do this sequentially ordered by \(\tau\), starting from 1). The resulting schedule is valid and the sum of sizes of \(f_e\)s is less than its value for original schedule.

Then, at each time step \(\tau_i^e\), starting from 1, we decrease \(\lambda_i^e\) to \(size(\sigma_i)\). The resulting schedule still satisfies the question and has a cost of at most \(F\).

Now we construct the solution of Fault Caching problem in this way:
\begin{gather*}
	b_t = \{pg_i | \forall i; \exists (t, \lambda) \in f_{e_i}, \lambda - c_{e_i}(t-1) > 0\} \\
	d_t = \{pg_i | \forall i; \exists (t, \lambda) \in f_{e_i}, \lambda - c_{e_i}(t-1) < 0\},
\end{gather*}
and the reduction is complete. We know that the Fault Caching is NP-complete  and it's reduced to scheduling problem above in polynomial time, thus, we can infer that scheduling problem with Batch input and step cost is NP-complete.
\end{proof}

\subsection{Online Input - Step Cost}
In this subsection, we study the problem of scheduling capacity changes on payment network with the step cost function and online input. With this type of input, in every time step \(i\) we just have \(t_j \in \Gamma; 1 \leq j \leq i\) and we don't have any information about the future transactions. In this context, we deal with online algorithms and competetive analysis. 

In competitive analysis, an online algorithm $A$ is compared to an optimal offline algorithm. An optimal offline algorithm knows the entire request sequence in advance and can serve it with minimum cost. Given a request sequence $\sigma$, let $C_A(\sigma)$ denote the cost incurred by $A$ and let $C_{OPT}(\sigma)$ denote the cost incurred by an optimal offline algorithm $OPT$. The algorithm $A$ is called $c$-competitive if there exists a constant $a$ such that: 
$$C_A(\sigma) \leq c.C_{OPT}(\sigma) + a, $$
for all request sequences $\sigma$. Here we assume that $A$ is a deterministic online algorithm. The factor $c$ is also
called the competitive ratio of $A$ (\cite{albers1996competitive}). 
\begin{theorem}
Given a graph \(G = (V,E) \) and a set of online transactions \(\Gamma\). There is no deterministic $c$-competitive online scheduling algorithm $f = \{f_e\}_{e \in E}$ such that \(c \leq \Delta\) (maximum degree of nodes in \(G\)) for finding the minimum cost modification schedule with the step cost function \(SC(G, \Gamma)\):
\begin{gather*}
	SC(G, \Gamma) = \sum_e | f_e|
\end{gather*}
\end{theorem}
\begin{proof}
Let $g$ be the optimal batch solution. For each deterministic algorithm $f$ we devise a network and a series of transactions such that:
\begin{gather*}
	\Delta \times SC_g(G, \Gamma) \leq SC_f(G, \Gamma)
\end{gather*}
Consider a graph \(G(V, E), V=\{v_1, v_2\}\) and \(E=\{e_1,e_2,...,e_{\Delta}\}\) where all \(e_i\)s are edges from \(v_1\) to \(v_2\). For algorithm \(f\) we have:
\begin{gather*}
	\sum_{e \in E}c_e^i \leq C_{v_1}(i); \; \forall i \in T \rightarrow \\
	\exists j: c^i_{e_j} \leq \frac{C_{v_1}(i)}{\Delta} + \epsilon; \forall i \in T, \forall \epsilon > 0
\end{gather*}
We devise \(t_i\) such that:
\begin{gather*}
	t_i = <v_1, v_2, e_j, \frac{C_{v_1}(i)}{\Delta} + \epsilon)
\end{gather*}
With this input for every \(i\) we have \((i,\frac{C_{v_1}(i)}{\Delta})\) in \(f_{e_j}\). Thus, for every time step we have at least one changes so:
\begin{gather*}
	SC_f(G, \Gamma) \geq n,
\end{gather*}
but for optimal solution, $g$, when we need an increase to \(e_j\) at time step i we decrease \(e_k\) such that \(e_k \in p_{i+\Delta}\). Thus, $g$ needs a change for every \(\Delta\) step. So:
\begin{gather*}
	\Delta \times SC_g(G, \Gamma) \leq SC_f(G, \Gamma)
\end{gather*}
\end{proof}

%% file: heuristics.tex
\section{Heuristic Algorithms}
In previous sections, we approached the problem of finding the minimum cost modification schedule for a given payment network and set of transactions in four different cases: Batch input-linear cost, Batch input-step cost, Online input-step cost, Online input-linear cost. 

For the batch input-step cost case, we proved that finding the optimum solution is NP-complete. This section will try some heuristic algorithms for this case, and we will compare their performances to find out the best heuristic. First of all, we will give a model of the problem of implementing heuristics for that. Then, we will introduce the heuristic algorithms we want to try. Finally, We will show empirical results and compare them.

\subsection{Model}
Suppose that we have a graph \(G\) with \(n\) channels and a set of transactions \(T\), over $G$, with \(m\) transactions. Again, according to the definition of the problem, each channel and node has a capacity. Node capacities cannot change, but channels' capacity can change with cost \(1\) for each modification. The sum of capacities of all channels of a specific node can't be greater than the capacity of that node. Now, assume that we have found a solution \(S\) for the problem.

The procedure of the algorithm will be in this way: At each step \(i\), $S $ will do some modifications on channels' capacities, and the cost will increase by the number of modifications. Then, transaction \(t_i\) of \(T\) will be executed, and the capacity of channels in the path of \(t_i\), the capacity of the source node, and capacity of target node will decrease by the value of \(t_i\). If limits of capacities cannot route the transaction, the solution will get some penalty, and the cost will increase by that, and \(t_i\) will be skipped. This procedure will repeat from \(i=0\) to \(i=m\).

We want to use an array to represent solution \(S\) and then try to change that array with heuristic algorithms to find better solutions. The first and the most straightforward way to represent \(S\) is to use a 2D matrix. The number in row \(i\) and column \(j\) of the matrix will be the value of modification on channel \(i\) in step \(j\). The cost of the solution will be the number of non-zero elements plus the penalty for the number of skipped transactions. 

However, there are some problems with this way of solution modeling. One problem is that the number of parameters in the solution matrix array is \(m \times n\), i.e., the number of channels times the number of transactions. This number of parameters will increase our heuristic algorithms' time and complexity and increase randomness in our algorithms.

Another problem is the way we represent the numbers in the matrix. If the numbers are absolute values of increase or decrease to capacities, those values may depend on every capacity and transaction's values. Also, these values have no bound and can change from minus infinity to plus infinity. These problems can make the progress of the algorithm hard and slow. However, if, in some way, we can use normalized coefficients from 0 to 1 instead of these absolute values, that can make our problem easier.

For finding a better way to represent a solution of an algorithm in an array, we will use some lemmas. Here we will discuss these lemmas.
\begin{lemma} 
Assume that in step \(i\) we have transaction \(t_i\) that uses path \(p_i\). Define \(N_i\) the set of all nodes in path \(p_i\). If two nodes \(v_1\) and \(v_2\) are not in \(N_i\), there is no need to change capacity of edge \(e\) between \(v_1\) and \(v_2\), in step \(i\).
\end{lemma}
This lemma is quite intuitive because we can push changes to irrelevant edges to the next steps, and the resulting cost will be less than or equal to the first solution.

\begin{lemma}
If a transaction \(t_i\) in step \(i\) has value \(v_i\) and there is channel \(e\) in path \(p_i\) of \(t_i\) and capacity of \(e\) is greater than \(v_i\), there is no need to change the capacity of \(e\) in step \(i\).
\end{lemma}
The intuition behind this lemma is just like before. We can push changes to edges that have enough capacity for now to the next steps.

Now we can provide another model for our solution \(S\). We will have an array of numbers between \(0\) and \(1\) to represent this solution. Every step a transaction needs to pass through a channel, we check that if the channel's capacity is enough for that value, we let it go and do not change the channel capacity. But, if the channel's capacity is not enough, there can be two states.

First, the source node's capacity is enough, but the channel's capacity is not enough. We increase the channel's capacity as much as it can pass the transaction; then, we increase it by the amount that it can increase multiplied by a coefficient from the solution array.

Another state is when both the channel's capacity and the source node's capacity are not enough. In this state, we need to decrease the capacity of other channels connected to the source node so that its capacity will increase, and then we can act like before. Therefore, we go through all other channels connected to the source node and multiply them by a coefficient from the solution array. Whenever the node capacity is enough for the transaction to pass, we change the channel's capacity like the former state.

There is a problem with this model, and that is the length of the array. If the array's length is too long, then there would be some parameters that are not useful. If the length of the array is too short, then we lack some coefficients for our solution. To tackle this problem, we assumed that if we ran out of coefficients in the solution, we use 0 as the coefficient. Also, we consider the length of the solution array as a hyper-parameter to our heuristic model that can be tuned, and different values can be checked to find the best one.

\subsection{Algorithms}
After providing a model for our algorithm and its solution, we will test some heuristic algorithms to find the best solution and compare their performance. Here we'll explain these algorithms.
\begin{itemize}
	\item \textbf{Alg1}. This algorithm is based on Genetic Algorithms \cite{yang2020nature}. First, we create random solution arrays to the number of population. Then, in each generation, we calculate the fitness function for each solution and find the best ones as parents. Then we do a crossover between those selected solutions and do some random mutations on resulting arrays. The resulting arrays after crossover and mutation plus the parent arrays will form the next population. This process will repeat until the algorithm converges or the maximum number of generations exceeds.
	\item \textbf{Alg2}. This algorithm is based on Random Hill Climbing. At each step, we move our solution to a new random state. That means we increase or decrease some random elements of our array by a constant step. After each move, we calculate the fitness function of our array, and if it has improved, we stay; otherwise, we step back to the previous solution. This process repeats until the maximum number of steps exceeds.
	\item \textbf{Alg3}. This algorithm is based on Late Acceptance Hill-Climbing \cite{BURKE201770}. Researches have shown the superiority of this algorithm over other optimization methods in many application fields. It is similar to the previous algorithm, but it has a history array. After each step, it checks the solution's fitness with the previous step and with fitness in specific steps before from our history array. For example, if we are in step 52 and the history array's length is 30, our algorithm compares the fitness of the next step with the previous one and the number in the index of 22 in our history array. If it is better than any of them, it will move to the next step. Moreover, if new fitness is better than the fitness in our history, it will replace it in the history array. This process repeats until the maximum number of steps exceeds.
	\item \textbf{Alg4}. This algorithm is based on the PSO algorithm \cite{yang2020nature}. It initiates a population of swarms with a position in search space and a velocity. At each step, this algorithm updates the velocity of each particle using the PSO well-known formula depending on its current position, its best-known position, and the best-known position of the entire swarm. Then, it updates the particle's position depending on this new velocity. If this new position's fitness is better than this particle's best-known position or the entire swarm's best-known position, updates them. This process repeats until the algorithm converges or the maximum number of generations exceeds.
	\item \textbf{Alg5}. This algorithm is based on Simulated Annealing \cite{delahaye2019simulated}. It initiates the algorithm with a state and a temperature. Each step moves the state (like movements we had in Alg3 or Alg2) and calculates its fitness. Then calculates the probability of acceptance of the new state with this formula: \(e^{(old\_fitness - new\_fitness)/T}\). Then, it compares acceptance probability with a random number, and if it is greater than that, it moves to the new state. After each step, the temperature decreases by being multiplied by a constant. This process repeats until the temperature goes below some specific number.
\end{itemize}

\subsection{Empirical Results}
In this subsection, we present the results of experimenting with the heuristics mentioned above on real-world data. 
For data of nodes and channels, we used \cite{bigsun} API and crawled data from that. This data was massive, So for making it smaller, we sampled a connected graph from it.

There is no open dataset for transactions over the Lightning network. So we generated a random transaction dataset by choosing some random paths in our graph and some random values from the uniform distribution. We used the number of changes in channel capacities as our objective function. 
For each algorithm, some hyper-parameters must be tuned—for example, the length of the solution array or the coefficient of temperature in simulated annealing. We ran many simulations and tests on algorithms to find the best hyper-parameters for each one of them.

The result of growth and convergence of algorithms can be seen in Figure \ref{heuristics_chart}. From this simulation, these results can be concluded:
\begin{itemize}
	\item The best overall efficiency was obtained from Alg3, which is based on the Late Acceptance Hill Climbing, with the second-best being Alg4, which is based on the classic PSO algorithm. Late Acceptance Hill Climbing has shown to be very effective in many other applications. 
	\item Alg4 and Alg5 converge after few steps, but Alg1, Alg5, and Alg2 took much more time to converge.
	\item Alg3 is not having improvements at the first steps, but it starts at the ending steps. Maybe it depends on the length of the history array in our Late Acceptance Hill climbing algorithm
	\item All the algorithms were stable and converged.
\end{itemize}

\begin{figure}[htp]
\centering
\includegraphics[width=\textwidth]{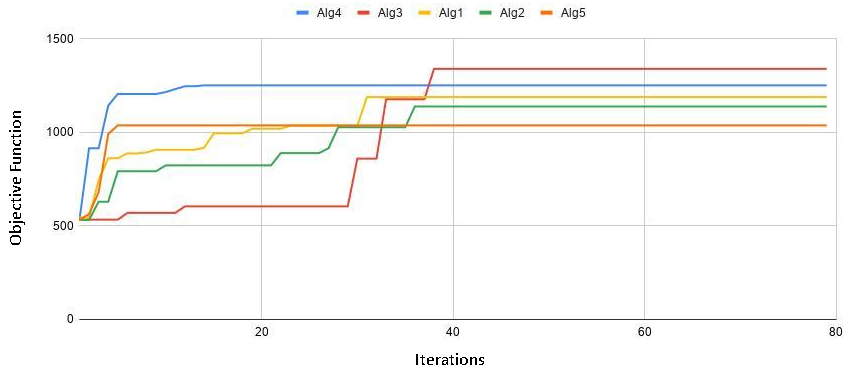}
\caption{Simulation results for the Stable network problem with 5 different algorithms.}
\label{heuristics_chart}
\end{figure}

%% file: conclusion.tex
\section{Conclusion}
We formalized the problem of capacity changes scheduling on payment networks while minimizing the cost of maintaining the network.

We have considered different types of this problem, depending on the input type of transactions and the cost function type. We proposed a polynomial-time algorithm for the batch input and the linear cost case for finding the optimal solution using linear programming. We proved the NP-completeness of finding the optimal modification scheme for the batch input and the step cost function. For the online input and step cost case, we proved no deterministic c-competitive online scheduling algorithm such that \(c \leq \Delta\), where \(\Delta\) is the maximum degree of the nodes in our payment network.

After that, we proposed some heuristic algorithms for the batch input and step cost case. We ran some simulations on those algorithms and compared them in an empirical experiment. This paper's results and algorithms can be used in all off-chain payment networks like the Lightning Network or the Raiden Network.

We have not addressed other cost functions. Moreover, we have not considered the general form of the problem, the problem in which a general function $f$ is assumed as a cost function.

For the heuristic part, we have not implemented heuristic algorithms for other cases, such as online input and step cost. Additionally, There are more heuristic algorithms we could investigate.

These are the problems that need to be addressed in the future.